\documentclass[11pt]{article}
\usepackage{amsmath}
\usepackage{amsfonts}
\usepackage{amsthm}
\usepackage{amssymb}
\usepackage{fullpage}
\usepackage{dsfont}
\usepackage{hyperref}
\usepackage{graphicx}
\usepackage{color}

\newcommand{\trace}{{\rm Tr}}

\newcommand{\set}[1]{{\left\{#1\right\}}}    % braces for set notation

\newcommand{\ve}[1]{\mathbf{#1}}
\newcommand{\abs}[1]{\left\lvert #1 \right\rvert}

\newcommand{\complex}{{\mathbb C}}
\newcommand{\reals}{{\mathbb R}}

\newcommand{\nats}{{\mathbb N}}

\def\ket#1{ | #1 \rangle}

\newcommand{\braket}[2]{\mbox{$\langle #1  | #2 \rangle$}}

\newcommand{\spa}[1]{\mathcal{#1}}

\newcommand{\class}[1]{{\rm #1}}

\newcommand{\poly}{\operatorname{poly}}

\newcommand{\n}[1]{\mathcal{N}(#1)}
\newcommand{\nc}[1]{\mathcal{N}(#1)^\perp}

\newcommand{\gtnit}{{\rm gTNZ}}
\newcommand{\tnit}{{\rm TNZ}}
\newcommand{\tnitp}{{\rm TNZ+}}

\newtheorem{theorem}{Theorem}

\newtheorem{lemma}{Lemma}

\newtheorem{defn}{Definition}

\newtheorem{cor}[theorem]{Corollary}
\newtheorem{problem}[theorem]{Problem}

\newtheorem{observation}{Observation}

\begin{document}
\title{Tensor network non-zero testing}
\author{
Sevag Gharibian\footnote{Simons Institute for Theoretical Computing, University of California, Berkeley, CA 94720, U.S.A.}$^{~,\dagger}$ \and Zeph Landau$^{\ast,\dagger}$ \and Seung Woo Shin\footnote{Electrical Engineering and Computer Sciences, University of California, Berkeley, CA 94720, U.S.A.} \and Guoming Wang$^\dagger$
}

\date{\today}

\maketitle
\begin{abstract}
Tensor networks are a central tool in condensed matter physics. In this paper, we study the task of tensor network non-zero testing (\tnit): Given a tensor network $T$, does $T$ represent a non-zero vector? We show that \tnit~is not in the Polynomial-Time Hierarchy unless the hierarchy collapses. We next show (among other results) that the special cases of \tnit~on non-negative and injective tensor networks are in NP. Using this, we make a simple observation: The commuting variant of the MA-complete stoquastic $k$-SAT problem on $D$-dimensional qudits is in NP for $k\in O(\log n)$ and $D\in O(1)$. This reveals the first class of quantum Hamiltonians whose commuting variant is known to be in NP for all (1) logarithmic $k$, (2) constant $D$, and (3) for arbitrary interaction graphs.
\end{abstract}

\section{Introduction}

One of the central aims of condensed matter physics is the study of ground spaces of local Hamiltonians. Here, a \emph{$k$-local Hamiltonian} is a sum $H=\sum_i H_i$ of Hermitian operators $H_i$, each of which act non-trivially on subsets of $k$ (out of a total of $n$) qudits. Such operators typically govern the evolution of quantum systems in nature, and in particular, their \emph{ground space} (i.e.~the eigenspace of $H$ corresponding to its smallest eigenvalue) characterizes the state of the corresponding quantum system at low temperature. Thus, the theoretical study of ground spaces of local Hamiltonians is crucial to understanding (e.g.) exotic phases of matter, such as superfluidity, which manifest themselves at low temperatures.

To this end, one of the key tools used by the condensed matter physics community is that of \emph{tensor networks} (see e.g.~Reference~\cite{CV09} for a survey). Specifically, tensor networks allow one to succinctly represent certain non-trivially entangled quantum states. As such, they play a crucial role in the study of ground spaces of local Hamiltonians. For example, in the early 1990's, White developed the celebrated DMRG heuristic~\cite{W92, W93}, which is nowadays recognized~\cite{OR95,RO97,VPC04, VMC08, WVSCD09} as a variational algorithm over 1D tensor networks known as Matrix Product States (MPS). The intuitive reason why DMRG works so well is that for 1D gapped Hamiltonians, the unique ground state turns out to be well-approximated by an MPS~\cite{Ha07}. Due in part to the success of DMRG, over the last two decades, a number of generalizations of MPS to higher dimensions have also been developed, such as Projected Entangled Pair States (PEPS)~\cite{VC04,VWPC06} and Multiscale Entanglement Renormalization Ansatz (MERA)~\cite{V07,V08}; such networks are able to represent larger classes of entangled states. Unfortunately, with this additional expressive power comes a price: Contracting an arbitrary tensor network is $\#$P-complete~\cite{SWVC07}. (Here, \emph{contracting} a network roughly means determining its value on a given input.)

\paragraph{Motivation.} Given that tensor networks play a fundamental role in condensed matter physics, and that contracting general networks is $\#$P-complete, here we ask a simpler question: \emph{Given a tensor network $T$, how difficult is it to decide whether $T$ represents a non-zero vector?}

Our original motivation for studying this question came from the following well-known open problem: Given a $k$-local Hamiltonian whose terms pairwise commute, what is the complexity of estimating its ground state energy? This is known as the commuting $k$-local Hamiltonian problem ($k$-CLH). Note that although asking for local terms to commute may intuitively make the problem seem ``classical'', such Hamiltonians can nevertheless have highly entangled ground states with exotic properties such as topological order~\cite{Kit03}.

For general $k$ and local dimension $d$, the best known upper bound on $k$-CLH is Quantum-Merlin-Arthur (QMA). However, the following special cases are known to be in NP: $k=2$ for local dimension $d\geq 2$~\cite{BV05}, $k=3$ with $d=2$ (as well as $d=3$ with a ``nearly Euclidean'' interaction graph)~\cite{AE11}, $k=4$ with $d=2$ on a square lattice~\cite{S11}, special cases of $k=4$ on a square lattice with $d$ polynomial in the number of qudits~\cite{H12_2}, and the case where the interaction graph is a good locally-expanding graph~\cite{AE13_1}. In particular, implicit in the approach of Schuch~\cite{S11} is a simple tensor network representation $T$ of the ground space of any commuting $k$-local Hamiltonian $H$; thus, the ability to verify in NP whether $T$ is non-zero would place $k$-CLH into NP for $k\in O(\log n)$ and $d\in O(1)$.

Finally, the question of whether a tensor network is non-zero has a fundamental physical interpretation: In 1967, Penrose showed~\cite{P67} that a spin network has non-zero norm if and only if it represents a physical situation allowed by the laws of quantum mechanics.

\paragraph{Results.} The decision problem we study in this paper is formally stated as follows. Below, $m$ denotes the number of physical edges in the network, each of which is assumed to have dimension $d$. (See Section~\ref{scn:def} for definitions.)

\begin{problem}[Generalized Tensor Network Non-Zero Testing (\gtnit)]\label{def:GTNIT}
    Given a classical description of a tensor network $T:[d]^m\mapsto\complex$ and threshold parameters $\alpha\geq \beta\geq0$ such that $\alpha-\beta\geq 1$,
    \begin{itemize}
        \item if there exists an input $x\in [d]^m$ such that $\abs{T(x)}\geq \alpha$, output YES, and
        \item if for all $x\in [d]^m$, $\abs{T(x)}\leq \beta$, output NO.
    \end{itemize}
\end{problem}

\noindent For convenience, we use the shorthand \tnit~to refer to \gtnit~with parameters $\alpha=1$ and $\beta=0$. Note that the key parameter here is $\beta=0$, and there is no loss of generality in setting $\alpha=1$. This is because in this paper, we assume the entries of the input tensor network $T$ are specified as complex numbers with rational real and imaginary parts. Since the value of $T$ on any input is given by a polynomial in the entries of the nodes with $d^n$ terms, it follows that the gap in any instance of \tnit~can be trivially amplified to $1$ by multiplying $T$ by an appropriate scalar based on the size of the network and the precision used to encode $T$.

Our main results are as follows.
\begin{enumerate}
    \item (Theorem~\ref{thm:hard}) \gtnit~is $\#P$-hard.
    \item (Theorem~\ref{thm:tnithard}) \tnit~is not in $\Sigma_i^p$ unless the Polynomial Hierarchy (PH) collapses to $\Sigma_{i+2}^p$. Here, $\Sigma_i^p$ denotes the $i$th level of PH (see Section~\ref{scn:def} for definitions).
    \item (Theorem~\ref{thm:tnit+}) \tnit~with the additional restriction that $T$'s nodes contain only non-negative entries is NP-complete, even when $T$ is given by a $3$-regular graph with edges of bond dimension $3$.
    \item (Theorem~\ref{thm:injnonzero} and Theorem~\ref{thm:noinj}) If $T$'s nodes represent injective linear maps, then $T$ is non-zero. Conversely, there exists a non-zero tensor network $T$ which does not have a ``geometrically equivalent'' injective tensor network representation $T'$. This implies that injective networks cannot exactly represent a state with long-range correlations (Observation~\ref{obs:longrange}).
    \item (Lemma~\ref{l:reduction} and Corollary~\ref{cor:stoq}) The commuting variant of the Stoquastic Quantum $k$-SAT problem is in NP for any $k\in O(\log n)$ and local dimension $D\in O(1)$.
\end{enumerate}

\noindent\emph{Remarks on previous work:} The non-commuting variant of the Stoquastic Quantum $k$-SAT problem is MA-complete~\cite{BT09}. Tensor network non-zero testing has been previously studied, e.g., in~\cite{MB12,JBCJ13,BMT14}, in related but distinct contexts. For example, Reference~\cite{MB12} shows that the following decision problem is undecidable: Given a set $F$ of quantum gates, does every network built using the gates in $F$ (where each gate is allowed one postselection) have non-zero norm? Another problem related to non-zero testing (in that it also concerns the physicality of a given tensor network) has been studied, for example, in Reference~\cite{KGE14}; the latter shows that determining whether a 1D translationally invariant Matrix Product Operator corresponds to a positive semidefinite operator can be either NP-hard or undecidable, depending on how the problem is phrased. Injective tensor networks have previously been studied in unrelated contexts, e.g., in the translationally invariant case in~\cite{PSGWC10}, or in the setting of $G$-injectivity~\cite{SCP10}.

\paragraph{Significance.} Although we do not fully resolve the complexity of the commuting local Hamiltonian problem ($k$-CLH), the strength of our approach is that, to the best of our knowledge, our line of attack on $k$-CLH is the first which does \emph{not} rely on Bravyi and Vyalyi's Structure Lemma~\cite{BV05}. In fact, it is purely this novel viewpoint which allows us to easily place the Stoquastic Quantum $k$-SAT problem into NP for any $k\in O(\log n)$ (Corollary~\ref{cor:stoq}). Moreover, although Theorem~3 suggests that testing whether an \emph{arbitrary} tensor network is non-zero is unlikely to be in NP, it is entirely plausible that the simple structure of the specific network $T$ which arises in the context of $k$-CLH (see Lemma~\ref{l:reduction}) \emph{can} be exploited to allow non-zero verification in NP.

Finally, as tensor networks are an important tool in condensed matter physics, it is crucial to understand their strengths and limitations. Result (2) shows that even the simple task of determining whether a given network $T$ represents a non-zero vector is in general very difficult. This underscores the need for cleverly designed classes of tensor networks such as MERA, which both manage to represent physically meaningful states, as well as allow efficient computation of local expectation values. To this end, we hope that our findings help guide the search for new key properties which make certain classes of tensor networks ``manageable''. For example, the fact that \tnit~on non-negative or injective networks lies in NP suggests that perhaps there are other physically relevant types of computations which can be performed on such networks ``easily'' (i.e.~in a complexity class below $\#$P).

\paragraph{Organization of this paper.} This paper is organized as follows. In Section~\ref{scn:def}, we formally define tensor networks and the Polynomial-Time Hierarchy. Section~\ref{sscn:hardness} shows complexity-theoretic hardness results for \tnit. In Section~\ref{sscn:easier}, we study easier special cases of \tnit~which fall into NP, such as non-negative and injective tensor networks. Section~\ref{scn:apps} discusses applications of \tnit~to the commuting local Hamiltonian problem. We conclude with open questions in Section~\ref{scn:conclusion}.

\paragraph{Notation.}
We define $[n] := \set{1,\ldots ,n}$. Let $\reals^+$ and $\nats$ denote the sets of non-negative real numbers and natural numbers, respectively. For operator $A:V\mapsto W$, let $\n{A}$ and $\nc{A}$ denote the null space of $A$ and the orthogonal complement of $\n{A}$, respectively. The notation $\spa{U}(V)$ denotes the set of unitary operators mapping $V$ to itself.

\section{Definitions}\label{scn:def}
In this section, we introduce definitions used throughout this article. We begin with a brief introduction to tensor networks, which are useful in condensed matter physics.

\paragraph{Tensor Networks.} There are two views of tensor networks we discuss here: The \emph{vector} and \emph{linear map} views. To introduce the first, we begin by thinking of a tensor $M(i_1,\ldots,i_k)$ simply as a $k$-dimensional array; given inputs $i_1$ through $i_k$, $M$ outputs a complex number. We call such an object $M:[d_1]\times\cdots\times[d_k]\mapsto\complex$ a \emph{$k$-dimensional tensor}, where $d_i\in\nats$. %For example, a $1$-dimensional tensor is a vector in $\complex^{d_1}$, and a $2$-dimensional tensor is a matrix acting on $\complex^{d_1}\otimes \complex^{d_2}$.
Given two tensors, it is possible to ``compose'' them by ``matching up'' certain inputs; this is called \emph{edge contraction}, and is best depicted via a simple but powerful graph theoretic framework, shown in Figure~\ref{fig:disp}~\cite{GHLS14}. In Figure~\ref{fig:disp}(a), the vertex corresponds to the tensor $M$, and each edge corresponds to one of the input parameters or \emph{indices} of $M$.
\begin{figure}[t]\centering
  \includegraphics[height=3cm]{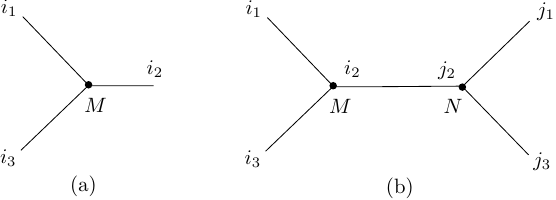}
  \caption{(a) A single tensor $M(i_1,i_2,i_3)$. (b) Two tensors $M(i_1,i_2,i_3)$ and $M(i_1,i_2,i_3)$ contracted on the edge $(M,N)$, yielding tensor $P(i_1,i_3,j_1,j_3)=\sum_{k}M(i_1,k,i_3)N(j_1,k,j_3)$.}\label{fig:disp}
\end{figure}
In Figure~\ref{fig:disp}(b), the edge $(M,N)$ denotes the contraction of $M$ and $N$ on their second index, the result of which is a $4$-dimensional tensor $P$ defined as
\[
    P(i_1,i_3,j_1,j_3)=\sum_{k}M(i_1,k,i_3)N(j_1,k,j_3).
\]
Since $P$ is $4$-dimensional, i.e.~has $4$ inputs, it is depicted as having four ``legs'' (i.e.~edges with only one endpoint) in Figure~\ref{fig:disp}(b).

By composing multiple tensors, we obtain a \emph{tensor network}.
\begin{figure}[t]\centering
  \includegraphics[height=3.8cm]{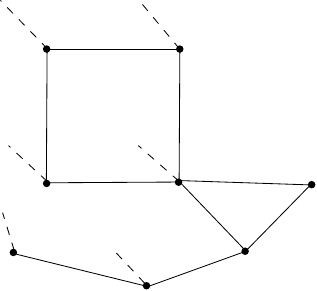}
  \caption{An arbitrary tensor network. The dashed edges denote physical edges, i.e.~inputs to the network, while the solid edges denote virtual edges, i.e.~contractions.}\label{fig:disp2}
\end{figure}
Figure~\ref{fig:disp2} depicts such a network. Here, open edges or legs are called \emph{physical} edges, whereas contracted edges are called \emph{virtual} edges. These names are physically motivated as follows. Recall that thus far, we have defined tensors as multi-dimensional arrays. The network $T$ in Figure~\ref{fig:disp2} is such an array taking in $6$ inputs $(x_1,\ldots,x_6)$; for each set of inputs, $T$ outputs a complex number $\alpha_{\ve{x}}$. The name \emph{vector} view now follows: $T$ can be thought of as representing a vector $\ket{v_T}$ such that given computational basis state $\ve{x}$, $T$ outputs amplitude $\alpha_{\ve{x}}$, i.e.~$\ket{v_T}=\sum_{x\in\set{0,1}^6}\alpha_x\ket{x}$. Why the names \emph{physical} and \emph{virtual} edges then? Typically in condensed matter physics, one thinks of the vertices in $T$ as corresponding to $d$-dimensional quantum systems. Then, each node of $T$ would have a physical edge of dimension $d$. The contracted edges, on the other hand, represent entanglement between systems; as such, they are called virtual edges. Their dimension $D$ is an important parameter known as the \emph{bond dimension} of the network.

Some further terminology: A network without physical edges is called a \emph{closed} network, and represents a complex number which can be computed by contracting the network. Given a closed network, a \emph{labeling} of its (virtual) edges means setting each index of every tensor to some fixed value, such that indices sharing an edge are set to the same value.

Next, we present the linear map view of tensor networks, which is perhaps best illustrated via the network $P$ in Figure~\ref{fig:disp}(b). In this view, rather than thinking of all $4$ physical edges as being \emph{inputs}, we can instead partition them into a set of inputs (say, edges $i_1$ and $i_3$) and a set of outputs (say, edges $j_1$ and $j_3$). Fix some values to inputs $i_1$ and $i_3$. Then, the result is a new network $P'$ with two remaining physical edges, $j_1$ and $j_3$. But $P'$ can now be thought of as a vector with inputs $j_1$ and $j_3$, just as in our first viewpoint! In other words, any input $(x,y)\in[d]\times[d]$ to $i_1$ and $i_3$ is mapped to a $d^2$-dimensional vector on inputs $j_1$ and $j_3$. By extending this action linearly over all basis states $(x,y)\in[d]\times[d]$, we have that $P$ acts as a linear map from inputs $i_1$ and $i_3$ to outputs $j_1$ and $j_3$, as claimed.

Finally, for clarity, we remark that in the definition of \gtnit~(Definition~\ref{def:GTNIT}), the input network is given by describing the connections between the nodes in the network, as well as the entries of each node/tensor (complex numbers are represented via a rational real and imaginary part).

\paragraph{The Polynomial-Time Hierarchy.} The Polynomial-Time Hierarchy (PH)~\cite{MS72} is defined as the union $\bigcup_i\Sigma_i^p$, where $\Sigma_i^p$ is defined as follows.
\begin{defn}[$\Sigma_i^p$]\label{def:PH}
    A decision problem $\Pi$ is in $\Sigma_i^p$ if there exists a polynomial time Turing machine $M$ such that given instance $x$ of $\Pi$,
    \[
        x\text{ is a YES instance } \Longleftrightarrow \exists y_1 \forall y_2 \exists y_3 \cdots Q_i y_i \text{ s.t. $M$ accepts }(x,y_1,\ldots,y_i),
    \]
    where $Q_i=\exists$ if $i$ is odd, and $Q_i=\forall$ if $i$ is even, and the $y_i$ are polynomial-length strings or \emph{proofs}.
\end{defn}

\section{Complexity of \tnit}\label{scn:complexity}

In this section, we show complexity-theoretic hardness of \tnit~(Section~\ref{sscn:hardness}), as well as study special cases of \tnit~which fall into NP (Section~\ref{sscn:easier}).

\subsection{Hardness of \tnit}\label{sscn:hardness}
Tensor networks are powerful objects; recall that simply contracting an arbitrary network $T$ is $\#$P-complete~\cite{SWVC07}. Thus, here we ask the natural question: Is \tnit~easier? For the general problem \gtnit, it is easy to answer this question in the negative using standard techniques by showing a polynomial time Turing reduction from the $\class{\#P}$-complete problem $\class{\#3COLORING}$, as we do now in Theorem~\ref{thm:hard} below. Recall that a \emph{$3$-coloring} of an undirected graph $G$ is an assignment of one of three colors to each vertex of $G$, such that neighboring vertices of $G$ are assigned distinct colors. In $\class{\#3COLORING}$, one is given a graph $G$ and asked to output the number of 3-colorings $M$ in $G$. Interestingly, it is known that $\class{\#3COLORING}$ is $\class{\#P}$-complete even when the degree of the graph is bounded by $4$ \cite{GJS74}, which will be important in the construction below. We remark that the construction here was used in~\cite{AL10} to show $\class{\#P}$-hardness of contracting tensor networks; the proof below simply appends a binary search to this construction, but we include a full description of the reduction for completeness.

\begin{theorem}\label{thm:hard}
    There exists a polynomial-time Turing reduction from $\class{\#3COLORING}$ to \gtnit.
\end{theorem}
\begin{proof}
    We first encode an arbitrary instance $G$ of $\class{\#3COLORING}$ into a (closed) tensor network $T$, such that contracting $T$ outputs the number of 3-colorings $M$, and then apply the standard idea of binary search to compute $M$ using a polynomial number of calls to a \gtnit~oracle.

	To construct $T$ from $G$, let $V$ and $E$ denote the sets of vertices and edges in $G$, respectively. For each vertex $v\in V$ and edge $e\in E$, we create nodes $T_v$ and $T_e$ in our tensor network, respectively. If vertex $v\in V$ is incident to edge $e\in E$, we connect $T_v$ and $T_e$ by an edge in the tensor network. Thus, the degree of $T_v$ is the number of edges that $v$ is incident to (which can be assumed to be at most $4$ \cite{GJS74}; this ensures $T_v$ has a constant size description), and the degree of $T_e$ is precisely $2$. All edges have bond dimension $3$. Next, we specify the action of $T$'s nodes. Each $T_v$ outputs $1$ if all inputs to this node agree, and outputs $0$ otherwise. This enforces $T_v$ to correspond to a consistent assignment to vertex $v$. On the other hand, each $T_e$ outputs $1$ if its two inputs disagree, and outputs $0$ otherwise. This enforces that the two incident vertices of $e$ have different colors. It is easy to see that the contraction of $T$ yields $M$, since each edge labeling of the network corresponding to a 3-coloring contributes $1$ to the sum.

    Given $T$, to now use an oracle for \gtnit~to compute $M$, note that for any positive integer $k$, solving \gtnit~on input $(T,k,k-1)$ allows us to determine if $M\geq k$ or $M \leq k-1$. Thus, since the number of colorings is at most $3^{\abs{V}}$, by invoking \gtnit~at most $\lceil \abs{V}\log 3\rceil$ times in conjunction with binary search, we can determine $M$ efficiently.
%
%    To thus see that for {any} positive $k$, \gtnit~indeed allows us to distinguish $M\geq k$ versus $M \leq k-1$, simply multiply each tensor in $T$ by the scalar $(2^t/k)^{1/(\abs{V}+\abs{C})}$ to obtain network $T'$. It follows that if $M\geq k$, then the contraction of $T'$ yields value at least $2^t$, whereas if $M\leq k-1$, then $T'$ yields value at most $2^t(k-1)/k$. Setting $\alpha=2^t$ and $\beta=2^t(k-1)/k$, we thus have our claim by using the fact that $k\leq 2^t$ to obtain that
%    \[
%        \alpha-\beta=2^t-\frac{2^t(k-1)}{k}=\frac{2^t}{k}\geq 1.
%    \]
\end{proof}

Theorem~\ref{thm:hard} tells us that general instances of \gtnit~are highly unlikely to be tractable. However, the proof relies critically on the ability to set the thresholds $\alpha$ and $\beta$ as needed. What if we fix $\alpha=1$ and $\beta=0$, i.e.~the case of \tnit? Clearly, the proof of Theorem~\ref{thm:hard} implies that this problem is at least NP-hard. Is it also in NP? The following theorem suggests not.

\begin{theorem}\label{thm:tnithard}
    If \tnit~is in $\Sigma_i^p$, then $\class{PH}\subseteq \class{P}^{\Sigma_{i+1}^p}$. In particular, this implies that the Polynomial Hierarchy collapses to the $(i+2)$-nd level.
\end{theorem}

To show Theorem~\ref{thm:tnithard}, we require two lemmas.

\begin{lemma}\label{l:add}
    Let $T$ be a closed tensor network on $n$ nodes and $m$ edges, where edge $i$ has bond dimension $d_i$ for $i\in[m]$, and such that the contraction of $T$ outputs value $M\in \complex$. Then, for any $N\in\complex$, one can construct in (deterministic) polynomial time a closed tensor network $T'$ satisfying the following properties:
    \begin{itemize}
        \item Contracting $T'$ outputs $M+N$, and
        \item $T'$ has $n$ nodes and $m$ edges, where edge $i$ has bond dimension $d_i+1$.
    \end{itemize}
\end{lemma}
\begin{proof}
     We construct $T'$ as follows. For any pair of nodes $v$ and $w$ in $T$ connected by edge $e=(v,w)$, we increase the bond dimension $d_e$ of $e$ by $1$; this extra dimension will play the role of a ``switch''. In particular, whenever $e$ is labeled with this ``switch'' value, we will say edge $e$ is set to SWITCH. To now describe how the vertices act on this extra dimension, fix some arbitrary node $v^*$, and relabel each node $v$ as $v'$. Then, in our new network $T'$, the action of each $v'$ is as follows:
    \begin{itemize}
        \item If all edges incident on ${v'}$ are not set to SWITCH, then ${v'}$ acts identically to $v$.
        \item Else, if there exists a pair of edges incident on $v'$, such that precisely one edge is set to SWITCH, then $v'$ outputs $0$.
        \item Else, if $v'=v^*$, then $v'$ outputs $N$. If $v'\neq v^*$, then $v'$ outputs $1$.
    \end{itemize}
    Thus, in $T'$ there are only two ways to label all edges to obtain a non-zero value. The first is when all edges are not set to SWITCH; contracting over all such labellings contributes value $M$ to the sum. The second is when all edges are set to SWITCH; in this case, $N$ is added to the sum. Thus, $T'$ outputs $M+N$, as desired.
%To finally obtain the statement of our claim, note that instead of adding a switch edge $e'$ for each $e$, we can increase the bond dimension of each $e$ by $1$ and have this extra dimension play the role of the switch as above. In this alternate construction, the vertex $v^*$ would instead output $N$ when all its incident edges simultaneously take on the ``switch'' bond value.\snote{maybe just use this version}
\end{proof}

\begin{lemma}\label{l:verifysharp}
    Given a graph $G$ with $n$ vertices and non-negative integer $k$, let $L$ denote the problem of deciding whether $G$ has at least $k$ 3-colorings. Then, $L\in\class{NP}^{\tnit}$.
\end{lemma}
\begin{proof}
    Let $O$ denote an oracle deciding \tnit. We construct a non-deterministic Turing machine $M$ with access to $O$ which decides $L$ in polynomial time. Suppose $G$ has $0\leq k^*\leq 3^n$ 3-colorings, for $n$ the number of vertices in $G$. Then, the action of $M$ on input $(G, k)$ is as follows:
    \begin{enumerate}
        \item Non-deterministically guess a value $k'$ satisfying $k\leq k'\leq 3^n$.
        \item As done in the proof of Theorem~\ref{thm:hard}, construct a tensor network $T$ encoding $G$, i.e.~whose contraction yields value $k^*$.
        \item Using Lemma~\ref{l:add}, map $T$ to a network $T'$ whose contraction yields value $k^*-k'$.
        \item Call $O$ on input $T'$. If $O$ outputs YES, output NO. Else, output YES.
    \end{enumerate}
    We now prove correctness. First, if we have a YES instance of $L$, then in step 1, $M$ guesses $k'=k^*$. The network $T'$ then yields value $k^*-k'=0$ upon contraction, signifying that we have guessed correctly. Thus, oracle $O$ outputs NO, in which case we flip the answer to YES in step 4. Conversely, if we have a NO instance of $L$, then any guess $k\leq k'\leq 2^n$ made by $M$ in step 1 will yield a network $T'$ whose value yields $\abs{k^*-k'}\geq 1$. Thus, oracle $O$ outputs YES in step 4, and we flip the answer to NO. To complete the reduction, note that each step above runs in non-deterministic polynomial time.
\end{proof}

With Lemmas~\ref{l:add} and~\ref{l:verifysharp} in hand, we now prove Theorem~\ref{thm:tnithard}.

\begin{proof}[Proof of Theorem~\ref{thm:tnithard}]
    Let $O_L$ be an oracle deciding language $L$ in the statement of Lemma~\ref{l:verifysharp}. Then, note that
    \begin{equation}
        \class{P}^{\class{\#3COLORING}}\subseteq \class{P}^{L}.\label{eqn:pl}
    \end{equation}
    Indeed, this holds since any call to a $\class{\#3COLORING}$ oracle can be simulated in polynomial time by applying binary search in conjunction with the oracle $O_L$. Now, if $\class{\tnit}\in\Sigma_i^p$, we have by Lemma~\ref{l:verifysharp} that
    \begin{equation}\label{eqn:pl2}
        \class{P}^{L}\subseteq\class{P}^{\class{NP}^{\tnit}}\subseteq \class{P}^{\class{NP}^{\Sigma_i^p}}=
        \class{P}^{\Sigma_{i+1}^p}\subseteq
        \class{NP}^{\Sigma_{i+1}^p}=\Sigma_{i+2}^p.
    \end{equation}
    On the other hand, since $\class{\#3COLORING}$ is $\class{\#P}$-complete, we have that
    \begin{equation}\label{eqn:pl3}
        \class{P}^{\class{\#3COLORING}}=\class{P^{\class{\#P}}}\supseteq \class{PH},
    \end{equation}
    where the last containment is given by Toda's theorem~\cite{T91}, which states that $\class{PH}\subseteq \class{P^{\class{\#P}}}$. Combining Equations~(\ref{eqn:pl}),~(\ref{eqn:pl2}), and~(\ref{eqn:pl3}), the claim follows.
\end{proof}

\subsection{Easier instances of \tnit}\label{sscn:easier}

In general, Theorem~\ref{thm:tnithard} implies that it is highly unlikely for \tnit~to lie in PH. In contrast, in this section, we study special cases of \tnit~whose complexity is provably in NP.

\noindent\paragraph{Non-negative tensor networks.} The first case we consider is very simple, and yet finds a nice application in Section~\ref{scn:apps}: The case in which the input tensor network's nodes contain only non-negative real numbers. Call such networks \emph{non-negative}. Then, defining \tnitp~as the problem \tnit~with a non-negative tensor network as input, we have the following.

\begin{theorem}~\label{thm:tnit+}
    \tnitp~is in NP, and is NP-hard even when the input network $T$ is given by a $3$-regular graph with all edges of bond dimension $3$.
\end{theorem}
\begin{proof}
    It is easy to see that \tnitp~is in NP; indeed, suppose we have a YES instance $T$, i.e.~there exists an input $x\in [d]^m$ such that $T(x)\geq 1$. Since all tensors comprising $T$ consist of non-negative entries, it follows that $T(x)\neq 0$ if and only if there exists a labeling of the tensors' virtual edges yielding a positive number. Such a labeling can be verified in polynomial-time, yielding the claim.

    Note now that the proof of Theorem~\ref{thm:hard} immediately yields that \tnitp~is NP-hard. However, the degree of the graph in that construction can be large. To obtain the statement of our claim here, we instead observe a many-one reduction from the NP-complete problem Edge-Coloring (ECOL) to \tnitp. Specifically, recall that in ECOL, one is given a simple graph $G=(V,E)$ and a choice of $c\in\nats$ colors, and asked whether there exists a coloring of the edges so that no two edges of the same color are incident on the same vertex. For this problem, our starting point is the fact that determining whether a simple $3$-regular graph is edge-colorable with $3$ colors is NP-hard~\cite{H81}. Thus, suppose $G$ is a simple $3$-regular graph. We construct a tensor network $T$ from $G$ as follows. For each vertex $v\in V$, create a tensor node $T_v:[3]^{3}\mapsto \set{0,1}$. For each edge $(u,v)\in E$, connect the tensor nodes $T_u$ and $T_v$ with an edge. Finally, define each tensor $T_v$ such that $T_v(x_1, x_2,x_3)=1$ if $x_1\neq x_2$, $x_2\neq x_3$, and $x_1\neq x_3$, and $T_v(x_1,x_2,x_3)=0$ otherwise. Note that this is a closed network which is $3$-regular, has bond dimension $3$ on all edges, and all tensor entries are non-negative.

    To finally see correctness, observe simply that each tensor $T_v$ acts as a ``local check'', such that $T_v$ outputs $1$ if and only if all its adjacent edges are given distinct values or \emph{colors}. Hence, the network evaluates to a non-zero value if and only if there exists a valid $3$-edge-coloring of $G$, i.e.~we have reduced the problem to an instance of \tnitp. As the reduction clearly runs in polynomial time, this completes the proof.
\end{proof}
 Theorem~\ref{thm:tnit+} shows that $3$-regular non-negative networks suffice to achieve NP-hardness for \tnit. In contrast, it is well known that tensor networks on $2$-regular graphs \emph{can} be efficiently contracted (even in the presence of arbitrary complex entries). This is because such graphs are a union of cycles and paths, and the latter two can be contracted similar to how Matrix Product States are contracted.
Finally, note that the proof of Theorem~\ref{thm:tnit+} also yields the following simple result.
\begin{observation}
    Contracting a non-negative, $3$-regular, planar tensor network with bond dimension $3$ is $\class{\#P}$-hard.
\end{observation}
\begin{proof}
    This follows simply because the contraction of the network constructed in the proof of Theorem~\ref{thm:tnit+} yields the number of valid edge-colorings of $G$. The latter problem is $\class{\#P}$-hard for $3$-regular planar graphs and $3$ colors~\cite{CGW14}.
\end{proof}

\noindent\paragraph{Injective tensor networks.} We now consider so-called \emph{injective} tensor networks, which were studied for example in the translationally invariant case in~\cite{PSGWC10}. To define such networks, we first require some terminology: Given a tensor network $T$ on vertex set $V$, let $S\subseteq V$. Then, the subnetwork of $T$ \emph{induced} by $S$ is the network consisting of all vertices in $S$, as well as all edges (physical and virtual) incident on vertices in $S$. An example is given by Figure~\ref{fig:disp3}.

\begin{figure}[t]\centering
  \includegraphics[height=3.5cm]{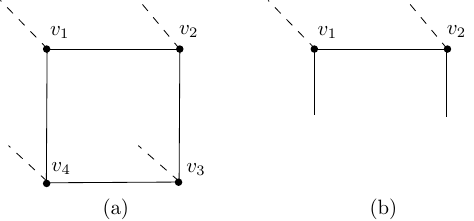}
  \caption{(a) A tensor network $T$. (b) The subnetwork of $T$ induced by vertices $\set{v_1,v_2}$.}\label{fig:disp3}
\end{figure}

\begin{defn}[Injective tensor network]\label{def:injective}
    Let $T:[d]^{n}\mapsto \complex$ be a tensor network. We call $T$ $k$-injective for $1\leq k \leq n$ if $T$ can be partitioned into $k$ sets of nodes $S=\set{S_{i}}_{i=1}^k$, such that for all $i$, the subnetwork $T_i$ of $T$ induced by $S_i$ has the following properties:

     \begin{enumerate}
        \item $T_i$ is connected.
        \item At least one node in $T_i$ has a physical edge.
        \item Let $L_i$ denote the linear map from the virtual edges crossing the cut $S_i$ versus $V\backslash S_i$ in $T$ (where $V$ is the vertex set of $T$) to the physical edges of $S_i$. Then, $L_i$ is an injective map.
     \end{enumerate}
\end{defn}

By exploiting the injective property of such networks, we can show the following.

\begin{theorem}\label{thm:injnonzero}
     If a tensor network $T$ is $k$-injective for some $k$, then $T$ is non-zero.
\end{theorem}
\begin{proof}
    Let $T:[d]^{n}\mapsto \complex$ be a $k$-injective tensor network, and let $S=\set{S_{i}}_{i=1}^k$ be a partition of the nodes of $T$ as in Definition~\ref{def:injective} with corresponding linear maps $L_i:(\complex^d)^{\otimes {n_i}}\mapsto(\complex^d)^{\otimes m_i}$. Now, since any $L_i$ is injective, it follows that the adjoint map $L_i^*$ from the physical to virtual edges is surjective. Thus, for each $L_i^*$, there exists an input $\ket{\psi_i}\in (\complex^d)^{\otimes m_i}$ to its physical edges such that the output along the virtual edges is the $n_i$-qudit product state $\ket{0}^{\otimes n_i}$. In other words, there exists a physical input to the network such that the contraction of the network along each edge involves only inner products of the form $\braket{0}{0}=1$. Thus, $T$ is non-zero, as claimed.
\end{proof}

An immediate corollary to Theorem~\ref{thm:injnonzero} is the following, which intuitively says: If a given tensor network $T$ is injective, then a prover can prove to a verifier that $T$ is non-zero simply by demonstrating that $T$ is injective (i.e. the prover does not \emph{a priori} know that $T$ is injective).

\begin{cor}\label{cor:inj}
    \tnit~for a $k$-injective network in which each of the $k$ sets of nodes in the injective partition are of size $O(\log n)$ is in NP. Here, $n$ is the number of nodes in the network, and we assume $d\in\Theta(1)$.
\end{cor}
\begin{proof}
    The prover here specifies the sets $S$ in Definition~\ref{def:injective}. The claim then follows by Theorem~\ref{thm:injnonzero} and the fact that a network of size $O(\log n)$ takes time $d^{O(n)}=\poly(n)$ to contract, thus allowing us to check whether each map $L_i$ specified by the prover is injective in polynomial time.
\end{proof}

Corollary~\ref{cor:inj} gives us an efficiently verifiable condition which can certify that a non-zero vector represented by tensor network $T$ is indeed non-zero. It is thus natural to ask whether a suitably defined converse of this statement might hold. For example, given a non-zero vector $\ket{\psi}$, does there always exist \emph{some} $k$-injective representation of $T$ in which the size of the sets $S_i$ are logarithmic? This question is interesting for two reasons. First, injective tensor networks are generic (see, e.g., \cite{PSGWC10}). Second, using the techniques in Section~\ref{scn:apps}, a positive answer to this question might be a step towards resolving the long-standing open question of whether the commuting $k$-local Hamiltonian for arbitrary $k\in\Theta(1)$ is in NP in the affirmative.

To make progress on this question, we define the notion of geometrically equivalent tensor networks. Specifically, we say that networks $T$ and $T'$ are \emph{geometrically equivalent} if the parameters of their underlying graphs (e.g.~number of nodes, placement of physical and virtual edges, physical dimension, bond dimension, etc\ldots) are identical. In other words $T$ and $T'$ differ only in the specifications of the tensors (i.e.~nodes) themselves. Note that the notion of geometric equivalence is arguably well-motivated, as often in Hamiltonian complexity, given a local Hamiltonian $H$ with interaction graph $G$, one fixes the geometry of the tensor network ansatz intended to represent the ground state of $H$ to match $G$.

With this definition in hand, we now show the following.

\begin{theorem}\label{thm:noinj}
    For all $k>2$, there exists a non-zero network $T$ which does not have a geometrically equivalent $k$-injective representation.
\end{theorem}
\begin{proof}
    We proceed by constructing a non-zero matrix product state (i.e.~1D tensor network) which satisfies the claim. To begin, consider the $n$-qubit state
    \[
        \ket{\psi}=\ket{0}\ket{0}^{\otimes n-2}\ket{0}+\ket{1}\ket{0}^{\otimes n-2}\ket{1},
    \]
    which can be represented by an MPS of bond dimension $2$ as follows. There are $n$ nodes in the network, which we label as $\set{v_i}_{i=1}^n$, where node $v_i$ corresponds to qubit $i$ of $\ket{\psi}$. Each node has a physical edge. Vertex $v_i$ is connected via a virtual edge to vertex $v_{i-1}$ if $i\geq2$ and to $v_{i+1}$ if $i\leq n-1$. The nodes $v_1$ and $v_n$ output $1$ if all their edges (i.e.~both physical and virtual) are labeled by $0$, or if all edges are labeled by $1$; otherwise, they output $0$. As for $v_{2}$ through $v_{n-1}$, these output $1$ if their physical edge is set to $0$ and either both virtual edges are $0$ or both are $1$; otherwise, they output $0$. Thus, the only edge labelings which produce a non-zero value are those with all edges labeled $0$, or when the physical edges are labeled $10^{n-2}1$ and the virtual edges are all labeled $1$. In both these cases, the network outputs $1$. Thus, the MPS represents $\ket{\psi}$, as claimed.

\begin{figure}[t]\centering
  \includegraphics[height=3.3cm]{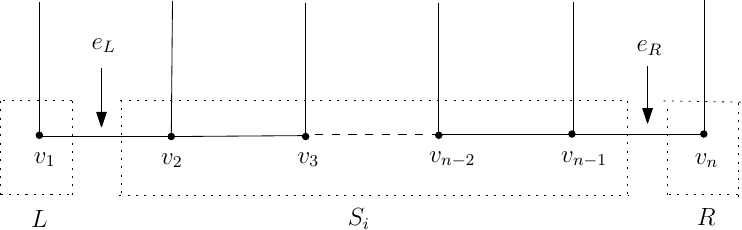}
  \caption{The network $T'$ in the proof of Theorem~\ref{thm:noinj}. In this example, $L=\set{v_1}$, $S_i=\set{v_2,\ldots,v_n}$, and $R=\set{v_n}$.}\label{fig:disp4}
\end{figure}

    Assume now, for sake of contradiction, that $\ket{\psi}$ admits a geometrically equivalent $k$-injective representation $T'$ for $k>2$. Since $k>2$, there exists a block $S_i$ such that $v_1,v_n\not\in S_i$. See Figure~\ref{fig:disp4} for an illustration. By definition of injective, we know that  $S_i$ is a contiguous set of nodes $\set{v_j,v_{j+1},\ldots,v_{m-1},v_m}$. Let $L=\set{v_1,\ldots,v_{j-1}}$ and $R=\set{v_{m+1},\ldots,v_n}$. We denote the virtual edges connecting $L$ and $R$ to $S_i$ as $e_L$ and $e_R$, respectively. Now, by definition of $\ket{\psi}$, if we input $0$ and $1$ on physical edges $1$ and $n$, respectively, $T'$ outputs $0$. Then, suppose the nodes in $L$ all receive physical input $0$, and the nodes in $R$ all receive physical input $0$, with the exception of $v_n$ which receives $1$. Let $\ket{\psi_L}$ and $\ket{\psi_R}$ denote the vectors output by $L$ and $R$ on the edges $e_L$ and $e_R$. Since the map corresponding to $S_i$ is injective, there exists a physical input to the nodes in $S_i$ such that $S_i$ outputs $\ket{\psi_L}$ on $e_L$ and $\ket{\psi_R}$ on $e_R$. But this implies $T'$ is non-zero on this input, which is a contradiction. This yields the claim.
\end{proof}

Note that the key idea behind the proof of Theorem~\ref{thm:noinj} is that given a fixed geometry for a tensor network, if one wishes to represent a quantum state with long-range correlations such as a Bell pair between the first and last qubits of a chain of tensors, injectivity can ``interfere'' with the ability of the first and last particles in the chain to correlate with one another. In other words, this proof technique leads to the following observation.

\begin{observation}\label{obs:longrange}
    Injective tensor networks cannot in general (exactly) represent a quantum state with long-range correlations (e.g.~such as a Bell pair between the first and last qubits of a chain of tensors).
\end{observation}
\noindent We remark that the condition of geometric equivalence plays an important role in this statement, as otherwise the notion of ``long-range'' is ill-defined. (In other words, to define ``long-range'', we assume the underlying physical systems are arranged according to some fixed geometry which is respected by the tensor network describing them.)

\section{Connections to Hamiltonian complexity}\label{scn:apps}

We now discuss connections between \tnit~and the commuting $k$-local Hamiltonian problem ($k$-CLH). Recall that in $k$-CLH, one is given a set of $k$-local Hermitian operators $\set{H_i}$, which act on $n$ $D$-dimensional qudits and which pairwise commute, as well as real parameters $\alpha$ and $\beta$ such that ${\beta-\alpha}\geq 1/\poly(n)$. We are asked to decide whether the smallest eigenvalue of $H=\sum_i H_i$ is at most $\alpha $ or at least $\beta$.

We first observe a connection between \tnit~and $k$-CLH. Specifically, we note that ground states of YES instances of $k$-CLH have a succinct tensor network description. Using this description, we then deduce that the ability to solve certain cases of \tnit~in NP would place $k$-CLH in NP for arbitrary $D\in O(1)$ and $k\in O(\log n)$. More generally, we have the following.

\begin{lemma}\label{l:reduction}
     For any $k\in O(\log n)$ and $D\in O(1)$, there exists a non-deterministic polynomial time mapping reduction from $k$-CLH on $D$-dimensional qudits to $\tnit$.
\end{lemma}
\begin{proof}
%    We use a setup similar to that of Schuch~\cite{S11}. Specifically, Let $(H=\sum_i H_i, \alpha, \beta)$ be an instance of $k$-CLH with ground state $\ket{\psi}$, and let $O$ be an oracle deciding \tnit. Since all $H_i$ pairwise commute, if we take a spectral decomposition $H_i=\sum_j\lambda_{ij}\Pi_{ij}$ of each $H_i$, it follows that for all $i$, there exists an eigenspace projector $\Pi_i:=\Pi_{ij}$ such that $\Pi_{ij}\ket{\psi}=\ket{\psi}$. Since all $\Pi_i$ also pairwise commute (as they all diagonalize in the same basis as $H$), it follows that the ground space of $H$ is given by $\Pi_H:=\prod_i\Pi_i$. We claim that $\Pi_H$ is the desired TNZ instance. If $H$ is a YES instance, the ground space is nonempty and therefore $\Pi_H$ must be a nonzero tensor network. On the other hand, if $H$ is a NO instance, the ground state must be empty and therefore $\Pi_H$ is zero.
%
    We use standard techniques and a setup similar to that of Schuch~\cite{S11}. Specifically, Let $(H=\sum_i H_i, \alpha, \beta)$ be an instance of $k$-CLH with ground state $\ket{\psi}$, and let $O$ be an oracle deciding \tnit. Since all $H_i$ pairwise commute, if we take a spectral decomposition $H_i=\sum_j\lambda_{ij}\Pi_{ij}$ of each $H_i$, it follows that for all $i$, there exists an eigenspace projector $\Pi_i:=\Pi_{ij}$ such that $\Pi_{ij}\ket{\psi}=\ket{\psi}$. Since all $\Pi_i$ also pairwise commute (as they all diagonalize in the same basis as $H$), it follows that the ground space of $H$ is given by projector $\Pi_H:=\prod_i\Pi_i$. In particular, either $\trace(\Pi_H)=0$ or $\trace(\Pi_H)\geq 1$, and the trace can be computed by contracting $i$th input qubit wire of $\Pi_H$ with the $i$th output qubit wire of $\Pi_H$; call this closed tensor network $T$. The reduction now proceeds as follows.
    \begin{enumerate}
        \item Non-deterministically guess projectors $\Pi_i$.
        \item Check that for each $i$, $\Pi_i$ encodes some eigenspace of $H_i$ with corresponding eigenvalue $\lambda_i$ such that $\sum_i\lambda_i \leq \alpha$. If not, reject.
        \item Feed $T$ into the oracle $O$ for \tnit~and return $O$'s answer.
    \end{enumerate}
	Correctness now follows easily, since for a NO instance of $k$-CLH, any projector $\Pi_H$ which is non-zero must satisfy the condition $\sum_i\lambda_i \geq \beta$; thus, either Step 2 or Step 3 will reject.
\end{proof}

Lemma~\ref{l:reduction} shows that if $\tnit\in\class{NP}$, then $k$-CLH is in NP for $k\in O(\log n)$. Unfortunately, we know from Theorem~\ref{thm:tnithard} that it is unlikely for \emph{arbitrary} instances of \tnit~to be solvable in NP. On the other hand, by exploiting the specific structure of the tensor network $T$ constructed in Lemma~\ref{l:reduction}, it may be possible to check whether $T$ is non-zero in NP. Here is a simple example for which this can be done --- the MA-complete Stoquastic $k$-SAT problem~\cite{BBT06}. In this problem, the input is a set of $k$-local orthogonal projection operators $\set{\Pi_i}$ whose entries in the standard basis are all non-negative, and the question is whether there exists a state $\ket{\psi}$ such that for all $i$, $\Pi_i\ket{\psi}=\ket{\psi}$.

\begin{cor}\label{cor:stoq}
    The variant of stoquastic quantum $k$-SAT in which all local projectors pairwise commute is in NP for any $k\in O(\log n)$ and $D\in O(1)$.
\end{cor}
\begin{proof}
        By definition of Stoquastic $k$-SAT, the network $T$ constructed in Lemma~\ref{l:reduction} for such a Hamiltonian has all real non-negative entries; thus, the claim follows from Theorem~\ref{thm:tnit+}.
\end{proof}

\section{Conclusions}\label{scn:conclusion}
In this paper, we have studied tensor network non-zero testing (\tnit). We have shown that \tnit~for {arbitrary} tensor networks is highly unlikely to be in the Polynomial-Time Hierarchy. We next obtained (among other results) that special cases of \tnit, such as non-negative and injective networks, lie in NP. Via a simple application of the non-negative case, we obtained that the commuting stoquastic quantum k-SAT problem is in NP for $k\in O(\log n)$ and $D$-dimensional systems for $D\in O(1)$.

Three questions we leave open are as follows. First, can the specific structure of the tensor network obtained in Lemma~\ref{l:reduction} be exploited to place the commuting $k$-local Hamiltonian problem into NP for $k\in O(\log n)$? Second, can the commuting stoquastic $k$-local Hamiltonian problem also be placed into NP using our techniques? Note that unlike for the stoquastic quantum k-SAT problem, here the local interaction terms are not necessarily projectors. Third, we have studied injective tensor networks --- how about the case of $G$-injective networks~\cite{SCP10}? Can anything interesting be said about \tnit~from a complexity theoretic perspective in this case?

\section{Acknowledgements}

We thank Jacob Biamonte for pointing out Penrose's~\cite{P67} interpretation of non-zero tensor networks to us, as well as for bringing to our attention existing works on tensor network non-zero testing. SG acknowledges support from a Government of Canada NSERC Banting Postdoctoral Fellowship and the Simons Institute for the Theory of Computing at UC Berkeley. ZL is supported by ARO Grant W911NF-12-1-0541, NSF Grant CCF-0905626 and Templeton Foundation Grant 21674. SWS is supported by NSF Grant CCF-0905626 and ARO Grant W911NF-09-1-0440. GW is supported by NSF Grant CCR-0905626 and ARO Grant W911NF-09-1-0440.

\bibliographystyle{alpha}
\bibliography{Sevag_Gharibian_Central_Bibliography_Abbrv,Sevag_Gharibian_Central_Bibliography}
\end{document}